\documentclass[11pt,twoside]{article}
\usepackage[dvips]{graphicx}
\usepackage{fancyhdr}
\usepackage{amsmath,amsfonts,epsfig,amsthm,amssymb}

\def\a{{\mathbf a}}

\def\u{{\mathbf u}}

\def\g{{\mathbf g}}

\def\h{{\mathbf h}}
\def\y{\mathbf{y}}
\def\x{\mathbf{x}}

\def\X{\mathbf{X}}
\def\P{\mathbf{\Phi}}

\def\e{\mathbf{e}}
\def\R{\mathbb{R}}

\def\A{\mathcal{A}}
\def\B{\mathcal{B}}
\def\H{\mathcal{H}}
\def\S{\mathcal{S}}

\newcommand{\refeq}[1] {equation (\ref{#1})}

\newtheorem{thm}{Theorem}
\newtheorem{lem}[thm]{Lemma}

\newtheorem{cor}[thm]{Corollary}

\begin{document}

%%%%%%%%%%%%%%%%%%%%%%%%%%%%%%%%%%%%%%%%%%%%%%%%%%
%
% The Title
%
\title{\bf\vspace{-39pt}Compressed Sensing with Nonlinear Observations and Related Nonlinear Optimisation Problems}

%%%%%%%%%%%%%%%%%%%%%%%%%%%%%%%%%%%%%%%%%%%%%%%%%%
%
% The Authors
%

\author{Thomas Blumensath \\ \small University of Oxford\\ \small Centre for Functional Magnetic Resonance Imaging of the Brain  \\
\small J R Hospital, Oxford, OX3 9DU, UK \\ \small tblumens@fmrib.ox.ac.uk\\}

%%%%%%%%%%%%%%%%%%%%%%%%%%%%%%%%%%%%%%%%%%%%%%%%%%
%
% Do not print the date
%
\date{}

%%%%%%%%%%%%%%%%%%%%%%%%%%%%%%%%%%%%%%%%%%%%%%%%%%
%
% Make the title and set the header styles to be
%   fancy for this page.  STSIP will adjust these
%   headings later.
%
\maketitle
\thispagestyle{fancy}

%%%%%%%%%%%%%%%%%%%%%%%%%%%%%%%%%%%%%%%%%%%%%%%%%%
%
% Setup the Headings for the article
%   - Please use all caps and initials for your
%     first and middle names
%
\markboth{\footnotesize \rm \hfill T. BLUMENSATH \hfill}
{\footnotesize \rm \hfill Nonlinear Compressed Sensing \hfill}

% Title.
% ------
%\title{Iterative Hard Thresholding beyond limits; \\ Guaranteed stability and improved performance}
%\title{A modified Iterative Hard Thresholding algorithm with guaranteed performance and stability}
%\title{Iterative Hard Thresholding in Practice}
%\title{A near optimal algorithm for inverse problems with a union of subspaces constraint.}
%\title{An efficient and near optimal algorithm for constrained inverse problems with a bi-Lipschitz property.}
%\title{}
%
% Single address.
% ---------------
%\author{Thomas Blumensath
%\thanks{Applied Mathematics, School of Mathematics, University of Southampton, University Road, Southampton, SO17 1BJ, UK}
%\thanks{\copyright\ This work might be submitted for possible publication. Copyright may be transferred without notice, after which this version may no longer be accessible.}
%}

%
%\markboth{Version: \today}{}

%\maketitle

%
\begin{abstract}
Non-convex constraints have recently proven a valuable tool in many optimisation problems. In particular sparsity constraints have had a significant impact on sampling theory, where they are used in Compressed Sensing and allow structured signals to be sampled far below the rate traditionally prescribed.  

Nearly all of the theory developed for Compressed Sensing signal recovery assumes that samples are taken using linear measurements. In this paper we instead address the Compressed Sensing recovery problem in a setting where the observations are non-linear. We show that, under conditions similar to those required in the linear setting, the Iterative Hard Thresholding algorithm can be used to accurately recover sparse or structured signals from few non-linear observations. 

Similar ideas can also be developed in a more general non-linear optimisation framework. In the second part of this paper we therefore present related result that show how this can be done under sparsity and union of subspaces constraints, whenever a generalisation of the Restricted Isometry Property traditionally imposed on the Compressed Sensing system holds. 
\vspace{5mm} 
\vspace{3mm}\\
%\noindent {\it 2000 AMS Mathematics Subject Classification} 65J22
\end{abstract}
%
%\begin{keywords}
%Sparse Data Modeling, Compressed Sensing, Iterative Hard Thresholding, Inverse Problems, Union of Subspaces
%\end{keywords}

\section{Introduction}
\label{section:intro}
Compressed Sensing \cite{candes06robust, donoho06formost, candes08cs} deals with the acquisition of finite dimensional sparse signals. Let $\x$ be a sparse vector of length $N$ and assume we sample $\x$ using $M$ \emph{linear} measurements. The $M$ samples can then be collected into a vector $\y$ of length $M$ and the sampling process can be described by a matrix $\P$. If the observations are noisy, then the Compressed Sensing observation model is
\begin{equation}
\y=\P\x+\e,
\end{equation}
where $\e$ is the noise vector. If $M<N$, then such a linear system is not uniquely invertible in general, unless we use additional assumptions on $\x$. Sparsity of $\x$ is such an assumption and Compressed Sensing theory tells us that, for certain $\P$, we can recover $\x$ from $\y$ even if $M<<N$, given that $\x$ has roughly $O(M)$ non-zero elements. However, in general, recovery of $\x$ is a combinatorial problem which is known to be NP-hard. Fortunately, under stricter conditions on $\P$, a range of different polynomial time algorithms can be used to recover $\x$ whenever $\x$ has roughly $O(M/log(N))$ non-zero elements.

One of the conditions that guarantees that we can use efficient algorithms is the \emph{Restricted Isometry Property} (RIP). A matrix $\P$ satisfies the \emph{Restricted Isometry Property} of order $2k$ \cite{candes06robust} if
\begin{equation}
(1-\delta) \|\x_1+\x_2\|^2 \leq \|\P(\x_1+\x_2)\|^2 \leq (1+\delta)  \|\x_1+\x_2\|^2
\end{equation}
for all $k$-sparse $\x_1$ and $\x_2$. The \emph{Restricted Isometry Constant} $\delta$ is defined as the smallest constant for which this property holds. One important interpretation of the RIP is in terms of the Lipschitz property of $\P$ and its inverse (where the inverse is defined only for sparse vectors and their image under $\P$) \cite{blumensath10subsp} and the condition states that, not only is $\P$ invertible on the set of sparse signals, this inverse is also smooth.

The RIP condition is a sufficient condition for the recovery of sparse $\x$. For example, \cite{candes08RIP} has shown that, for any $\x$, given an observation $\y=\P\x+\e$, where $\P$ has the Restricted Isometry Property with $\delta<\sqrt{2}-1$, then the solution $\x^\star$ to the convex optimisation problem 
\begin{equation}
\min_{\tilde\x} \|\tilde\x\|_1 \ : \ \|\y-\P\tilde\x\|_2 \leq \|\e\|_2
\end{equation}
has an error bounded by
\begin{equation}
\|\x^\star-\x\| \leq  c k^{-0.5} \|\x-\x_{k}\|_1 + c' \|\e\|,
\end{equation}
where $\|\cdot\|_1$ is the vector 1 norm, $\x_k$ is the best $k$ term approximation to $\x$ and where $c$ and $c'$ are two constants depending only on $\delta$.

Similar results have been obtained for other algorithms, such as the Compressed Sampling Matching Pursuit (CoSaMP) and Subspace Pursuit (SP) algorithms \cite{needell08COSAMP, dai08subspace} and the Iterative Hard Thresholding (IHT) algorithm \cite{blumensath09IHT}.

Several generalisations to this now classical Compressed Sensing setup have been introduced over the years. For example, some of the recovery algorithms, such as CoSaMP, SP and IHT, can be adapted to allow signals $\x$ to lie in a much more general, \emph{non-convex} constraint set $\A$. A powerful model here is for example the Union of Subspaces model, in which $\x$ is assumed to lie on one of several linear subspaces $\A_i$, though it is not known a priori on which subspace we are to look. Not only does this framework include the standard sparse model as a special instance, many other models of interest, such as analogue Compressed Sensing methods \cite{mishali09blind}, low rank matrix models \cite{goldfarb10convergence}, or structured sparse models \cite{baraniuk09model}, are also covered.

In this more general setting, with a general \emph{non-convex} constraint sets $\A$, Compressed Sensing can be formulated as the following optimisation problem,
\begin{equation}
\mathrm{argmin}_{\x\in\A} \|\y-\P\x\|_2^2,
\end{equation}
that is,  we search a vector $\x$ from the \emph{non-convex} constraint set $\A$ that minimises the sum of squares observation error. 

In this paper we look at a much more general setting, where we want to find the following optimum.
\begin{equation}
\label{equation:nonlinopt} 
\mathrm{argmin}_{\x\in\A} f(\x),
\end{equation}
where $f(\x)$ is now a much more general non-linear function of $\x$.

Of particular interest to us are \emph{non-linear Compressed Sensing} problems where $f(\x)=\|\y-\P(\x)\|$, with $\P(\x)$ being a non-linear mapping from one vector space to another. We address this non-linear Compressed Sensing problem first, however, the more general problem in \refeq{equation:nonlinopt} is of independent interest and an alternative treatment will be presented in the second part of this paper.

When we started studying these problems, not much was known of this general setting. However, since the first draft of this paper \cite{blumensath10nonlin}, similar ideas have been put forward independently in \cite{Xu11nonlin}, where the non-linear Compressed Sensing problem was tackled using a convexification approach, and in \cite{Bahmani12nonlin}, where non-convex optimisation problems were studied using an alternative greedy approach to the one discussed here. Whilst the first part of this paper contains more recent results, the second part of this paper is basically the same material that can be found in the earlier draft of this paper \cite{blumensath10nonlin}.

%%%%%%%%%%%%%%%%%%%%%%%%%%%%%%%%%%%%%%%%%%%%%%%%%%%%%%%%%%%%%%%%%%%%%%%%%%%%%%%%%%%%%%%%%%%
%%%%%%%%%%%%%%%%%%%%%%%%%%%%%%%%%%%%%%%%%%%%%%%%%%%%%%%%%%%%%%%%%%%%%%%%%%%%%%%%%%%%%%%%%%%
%%%%%%%%%%%%%%%%%%%%%%%%%%%%%%%%%%%%%%%%%%%%%%%%%%%%%%%%%%%%%%%%%%%%%%%%%%%%%%%%%%%%%%%%%%%

\section{Non-Linear Compressed Sensing}

%%%%%%%%%%%%%%%%%%%%%%%%%%%%%%%%%%%%%%%%%%%%%%%%
%\subsection{Non-linear measurements}

We are here interested in the development of a better understanding of what happens to the Compressed Sensing recovery problem when a signal is measured with some non-linear system. In particular, the hope is that, if the system is not too non-linear, then recovery should still be possible under similar assumption to those made in linear Compressed Sensing. To see the intuition behind why this might work, it is worth pointing out that in the linear setting, Compressed Sensing recovery works exactly in those cases in which the observation system is a bi-Lipschitz embedding. This means that, both, the observation mapping itself, as well as its inverse are Lipschitz functions. Obviously, these functions are only Lipschitz on the constraint set $\A$ and its image $\P\A$. In the linear setting, if $\P$ is bounded (e.g. in finite dimensional spaces), then $\P$ itself is obviously Lipschitz. The idea is now that, if Compressed Sensing works if both forward and backward maps are Lipschitz, maybe we can move away from a linear setting, where $\P$ is linear, and instead assume $\P$ to be Lipschitz, but non-linear.

The study of non-linear observation systems is not only of academic interest but has important implications for many real-world sampling systems, where measurement system can often not be designed to be perfectly linear.
Assume therefore that our measurements are described by a nonlinear mapping $\P(\cdot)$ that maps elements of the normed vector spaces $\H$ into the normed vector spaces $\B$. The observation model is therefore  
\begin{equation}
\y=\P(\x)+\e,
\end{equation}
where $\e\in \B$ is an unknown but bounded error term. Both $\H$ and $\B$ are assumed to be Hilbert spaces.

%%%%%%%%%%%%%%%%%%%%%%%%%%%%%%%%%%%%%%%%%%%%%%%%
\subsection{The Constraints}

As in Compressed Sensing, the interesting case occurs whenever the sampling system $\P$ is non-invertible or ill-conditioned. To cope with this, additional constraints need to be imposed on $\x$. Again, in the interest of generality, instead of restricting our discussion to sparse signals (however these might be defined in a general Hilbert spaces) we here use the more general framework of \cite{blumensath10subsp}  and assume that $\x$ lies in or close to a known set $\A$, where $\A\subset \H$ is a non-convex subset of $\H$. 
Of particular interest will be constraint sets $\A$ that can be described as the union of several subspaces. For these models we can write 
\begin{equation}
\A=\bigcup \A_i,
\end{equation}
where we use arbitrary closed subspaces $\A_i\subset\H$
 
One approach to recover $\x$ from $\y$ would be to mirror Compressed Sensing ideas and to define a convex objective function which can then be optimised using standard tools. However, for our general setup, it is not clear how this could be done. Instead, we use the Iterative Hard Thresholding (IHT) algorithm. To define this for general constraint sets $\A$, we again replace the hard thresholding step with a more general map which can be understood as a form of projection \cite{blumensath10subsp}.
Let $P_{\A}$ be a map from $\H$ to $\A$ such that  
\begin{equation}
\label{definition:projection}
\x_\A = P_{\A}(\x) \  : \  \x_\A\in\A, \  \|\x-\x_\A\|^2\leq\inf_{\hat{\x}\in\A}\|\x-\hat{\x}\|^2 + \epsilon.
\end{equation}
In this definition we have introduced an arbitrarily small constant $\epsilon> 0$, as there might not exist an $\x_{opt}$, such that $\|\x-\x_{opt}\|^2=\inf_{\hat{\x}\in\A}\|\x-\hat{\x}\|^2$. However, for simplicity, we will assume for the rest of this paper that $\A$ is a so called proximal set, which is just a fancy way of saying that the required optimal points indeed lie in the set $\A$, so that we use $\epsilon=0$ here. Nevertheless, it is easy to adapt our theory to the more general setting.

Note that this "projection" might not be defined uniquely in general, as for a given $\x$, there might be several elements $\x_\A$ that satisfy the condition in (\ref{definition:projection}). However, all we require here is that the map $P_{\A}(\x)$ returns \emph{a single element} from the set of admissible $\x_\A$ (which is guaranteed to be non-empty \cite{blumensath10subsp}). How this selection is done is of no consequence for our arguments here.

It is further worth noting that the relaxation offered by an $\epsilon>0$ in the definition of the above projection has also a computational advantage. Instead of having to compute exact optima, which for many problems are often difficult to find, many approximate algorithms can be used instead (see \cite{cevher11clash} for a more detailed discussion). 

%%%%%%%%%%%%%%%%%%%%%%%%%%%%%%%%%%%%%%%%%%%%%%%%
%\subsection{A union of subspaces signal model}

%%%%%%%%%%%%%%%%%%%%%%%%%%%%%%%%%%%%%%%%%%%%%%%%
%\subsection{Projections onto the Constraints}

%%%%%%%%%%%%%%%%%%%%%%%%%%%%%%%%%%%%%%%%%%%%%%%%
\subsection{The Iterative Hard Thresholding Algorithm for Non-Linear Compressed Sensing}
\label{subsection:algo}

For the linear Compressed Sensing problem, the Iterative Hard Thresholding (IHT) algorithm uses the following iteration
\begin{equation}
\x^{n+1}=P_\A(\x^n+\mu \P^*(\y-\P\x^{n}),
\end{equation}
where $\P$ is the linear measurement operator. 

In the non-linear case, let us approximate $\P(\x)$ using an affine Taylor series type approximation around a point $\x^\star$, so that $\P(\x)\approx\P(\x^\star)+\P_{\x^\star}(\x-\x^\star)$, where $\P_{\x^\star}$ is a linear operator (such as the Jacobian of $\P(\x)$, evaluated at point $\x^\star$). The matrix $\P_{\x^\star}$ thus will depend on $\x^\star$ in general. At iteration $n$ we then write the IHT algorithm as
\begin{equation}
\x^{n+1}=P_\A(\x^n+\mu\P_{\x^n}^*(\y-\P(\x^{n})).
\end{equation}
Indeed, as we show below in \ref{proof:1}, this algorithm can recover $\x$ under similar condition to those required from the IHT algorithm in the linear setting. All we require is that the matrices $\P_{\x^\star}$ satisfy a Restricted Isometry Property and that the error introduced in the linearisation is not too large, i.e. that $\|\P(\x_\A)-\P(\x_\x^n)-\P_{\x_\A}(\x_\A-\x^n)\|$ is small for large $n$.
\begin{thm}
\label{thm:main1}
Assume that $\y=\P(\x)+\e$ and that $\P_{\x^\star}$ is a linearisation of $\P(\cdot)$ at $\x^\star$ so that the Iterative Hard Thresholding algorithm uses the iteration $\x^{n+1}=P_\A(\x^n+\mu\P_{\x^{n}}^*(\y-\P_{\x^{n}}\x^{n})$.
Assume that $\P_{\x^\star}$ satisfies RIP
\begin{equation}
\alpha\|\x_1-\x_2\|_2^2 \leq \alpha\|\P_{\x^\star}(\x_1-\x_2)\|_2^2 \leq \beta \|\x_1-\x_2\|_2^2 
\end{equation}
for all $\x_1,\x_2, \x^\star\in\A$, with constants satisfying $\beta\leq 1/\mu < 1.5\alpha$. 
Define
\begin{equation}
\e_\A^{n}=\y-\P(\x^n)-\P_{\x^n}( \x_\A -\x^n )
\end{equation}
and  $\epsilon^k = b \sum_{n=0}^{k-1} a^{k-1-n} \|\e_\A^{n}\|^2 $, where $b=4/\alpha$ and $a=2/(\mu\alpha)-2$,
%and let $\e_\A^{n}=\y-\P_{\x^n}\x_\A$, 
then after 
\begin{equation}
k^\star = \left\lceil 2 \frac{\ln(\delta\frac{\sqrt{\epsilon^k}  }{\|\x_\A\|})}{\ln(2/(\mu\alpha)-2)} \right\rceil
\end{equation}
iterations we have
\begin{equation}
\|\x-\x^{k^\star}\| \leq (1+\delta)\sqrt{\epsilon^k} + \|\x_{\A}-\x\|.
\end{equation}
\end{thm}

Obviously, for the above theorem to make sense, we would require the error term $\epsilon^k$ to be well behaved. This is true whenever  $\|\y-\P(\x^n)-\P_{\x^n}( \x_\A -\x^n) \|_2$ is bounded, as then $\epsilon^k\leq b \sum_{n=0}^{k-1} a^{k-1-n} C$, for some constant $C$ so that the requirement that $1/\mu< 1.5\alpha$ ensures that $a=2/(\mu\alpha)-2<1$, which in turn implies that the geometric series $\sum_{n=0}^{k-1} a^{k-1-n} $ is bounded.

Indeed, if $a<1$ and if we can show that $\epsilon_n=\|\e_\A^{n}\|^2$ is bounded and convergent to some $\epsilon_{lim}$, then $\epsilon^k$ will also be bounded as the following argument shows
\begin{eqnarray}
\epsilon^k /b &=& \sum_{n=0}^{k-1}a^{k-n-1}\epsilon_n \nonumber\\
&=& \sum_{n=0}^{p-1}a^{k-n-1}\epsilon_n + \sum_{p}^{k-1}a^{k-n-1}\epsilon_n  \nonumber\\
&\leq& \sum_{n=0}^{p-1}a^{k-n-1}\epsilon_n + \epsilon_p \sum_{p}^{k-1}a^{k-n-1}   \nonumber\\
&\leq& \sum_{n=0}^{p-1}a^{k-n-1}\epsilon_n + \epsilon_p \frac{1}{1-a}\nonumber\\
&=& a^{k-p-1} \sum_{n=0}^{p-1}a^{p-n-1}\epsilon_n + \epsilon_p \frac{1}{1-a}\nonumber\\
&\leq& a^{k-p-1} \epsilon_0 \sum_{n=0}^{p-1}a^{p-n-1} + \epsilon_p \frac{1}{1-a} \nonumber\\
&\leq& a^{k-p-1} \frac{\epsilon_0}{1-a}  + \frac{\epsilon_p }{1-a} 
\end{eqnarray}
Thus, if we let $k$ and $p$ increase to infinity such that $0<k-p\rightarrow\infty$, then the first term on the left converges to zero whilst the second term converges to a limit depending on $\epsilon_{lim}$, so that, if we iterate the algorithm long enough, then 
\begin{eqnarray}
\lim_{k\rightarrow\infty} \epsilon^ k&\leq& \epsilon_{lim}\frac{b }{1-a} 
\end{eqnarray}
and the error term converges to
to \begin{equation}
\|\x-\x^\star\| \leq \sqrt{\epsilon_{lim}\frac{b }{1-a} } + \|\x_{\A}-\x\|.
\end{equation}

Actually, as shown in \ref{proof:2}, more can be said if we can establish the following bound for $\P(\x)$ and its linearisation $\|\P(\x_1)-\P(\x_2) - \P_{\x_1}(\x_1-\x_2) \|_2^2 \leq C \|\x_1-\x_2\|_2^2$.
\begin{cor}
\label{col:1}
Assume that $\y=\P(\x)+\e$ and that $\P_{\x^\star}$ is a linearisation of $\P(\cdot)$ at $\x^\star$ so that the Iterative Hard Thresholding algorithm uses the iteration $\x^{n+1}=P_\A(\x^n+\mu\P_{\x^{n}}^*(\y-\P_{\x^{n}}\x^{n})$.
Assume that $\P_{\x^\star}$ satisfies RIP
\begin{equation}
\alpha\|\x_1-\x_2\|_2^2 \leq \|\P_{\x^\star}(\x_1-\x_2)\|_2^2 \leq \beta \|\x_1-\x_2\|_2^2 
\end{equation}
for all $\x_1,\x_2, \x^\star\in\A$, and assume $\P(\x)$ and $\P_\x$ satisfy
\begin{equation}
\|\P(\x_1)-\P(\x_2) - \P_{\x_1}(\x_1-\x_2) \|_2^2 \leq C \|\x_1-\x_2\|_2^2, 
\end{equation}
with constants satisfying $\beta\leq 1/\mu < 1.5\alpha-4C$, 
then the algorithm converges to a solution $\x^\star$ that satisfies
\begin{equation}
\|\x-\x^\star\| \leq c   \|\e_\A\| + \|\x_{\A}-\x\|,
\end{equation}
where $\e_\A=\y-\P(\x_\A)$ and $c=\frac{2}{0.75\alpha-1/\mu-2C}$.
\end{cor}

\subsection{Example}

Before we proof Theorem \ref{thm:main1} and Corollary \ref{col:1}, let us give a simple example that shows how the above method and theory can be applied in a particular setting. Assume we have constructed a Compressed Sensing system, where a sparse signal $\x\in\R^N$ is measured using a \emph{linear} measurement system $\overline{\P}$. Also assume that we have constructed the system so that $\overline{\P}$ satisfies the Restricted Isometry Property with constants ${\alpha}$ and ${\beta}$. Now unfortunately, the sensors we have available for the actual measurements are not exactly linear but have a slight non-linearity, so that our measurements are of the form
\begin{equation}
\y=\overline{\P}f(\x)+\e,
\end{equation}
where $f(\cdot)$ is a non-linear function applied to each element of the vector $\x$. For simplicity, we will write $f(\x)=\x+h(\x)$, where again $h(\x)$ is a function applied element wise. We then have $f'(x) = 1+ h'(x)$.

It is not difficult to see that the Jacobian of $\overline{\P}f(\x)$ can be written as
\begin{equation}
\P_{\x^\star}=\overline{\P}+\overline{\P}H'_{\x^\star},
\end{equation}
where $H'_{\x^\star}$ is the diagonal matrix with the elements $ h'(\x^\star)$ along the diagonal.

To use Corollary \ref{col:1}, we thus need to determine a) the RIP constant of $\P_{\x^\star}$ and b) bound
 $\|\P(\x_1)-\P(\x_2) - \P_{\x_1}(\x_1-\x_2) \|_2^2$ as a function of $\|\x_1-\x_2\|^2$.
 
 The RIP constants are bounded for our example as follows.
 Assume that $|h'(x)|\leq M<1$, that $\x_1,\x_2\in\A$ and that $\overline{\P}$ satisfies the RIP with constants $\alpha$ and $\beta$. We then have
 %Let $H_{\x^\star}$ is the diagonal matrix with the elements $h_{\x^\star}$ along the diagonal.
 \begin{eqnarray}
 &&(\alpha^{1/2}1-\beta^{1/2}M) \|\x_1-\x_2\| \nonumber\\
&\leq& \|\overline{\P}(\x_1-\x_2)\| - \| \overline{\P}H'_{\x^\star}(\x_1-\x_2) \| \nonumber \\
 &\leq& \|\P_{\x^\star}(\x_1-\x_2)\| \nonumber\\
 &=& \|\overline{\P}(\x_1-\x_2) + \overline{\P}H'_{\x^\star}(\x_1-\x_2) \|  \nonumber\\
 &\leq& \|\overline{\P}(\x_1-\x_2)\|+ \| \overline{\P}H'_{\x^\star}(\x_1-\x_2) \| \nonumber \\
 &\leq& \beta\|(\x_1-\x_2)\|+ \beta\| H'_{\x^\star}(\x_1-\x_2) \| \nonumber\\
 &\leq& \beta^{1/2}(1+M) \|\x_1-\x_2\|, 
 \end{eqnarray}
 which proofs the following Lemma.
 \begin{lem}
Let $\P(\x)=\overline{\P}f(\x)$, where the function $f(\x)=x+h(\x)$ is applied element wise and where the derivative $h'(x)$ is absolutely bounded $|h'(\cdot)|\leq M$. Also assume that the matrix $\overline{\P}$ satisfies the RIP condition with constants $\alpha$ and $\beta$ for a set $\A$, then the matrix $\overline{\P}(\mathbf{I}+H'_{\x^\star})$ satisfies RIP with constants $(\alpha^{1/2}1-\beta^{1/2}M)^2$ and $\beta(1-M)^2$.
 \end{lem}

Let us now turn to point b). We have the bound 
 \begin{eqnarray}
&& \|\P(\x_1)-\P(\x_2) - \P_{\x_1}(\x_1-\x_2) \|_2^2 \nonumber \\
&=&  \|\overline{\P}h(\x_1)-\overline{\P}h(\x_2)-\overline{\P}H'_{\x^\star}(\x_1-\x_2) \|_2^2 \nonumber \\
&=&  \|\overline{\P} ( h(\x_1)-h(\x_2)-H'_{\x_1}\x_1+H'_{\x_1}\x_2) \|_2^2\nonumber \\
&\leq& \beta  \|( h(\x_1)-h(\x_2)-H'_{\x_1}\x_1+H'_{\x_1}\x_2)) \|_2^2,
 \end{eqnarray}
 where in the last inequality we assume $\overline{\P}$ to satisfy the RIP property and that $h(0)=0$ (Note that if we do not assume that $h(0)=0$, then the same reults still hold, though we have to replace $\beta$ by the operator norm of $\overline{\P}$). 
 Let us introduce the function $d_{\x^\star}(\x)=h(\x)-H'_{\x^\star}\x$, so that 
  \begin{eqnarray}
 \|\P(\x_1)-\P(\x_2) - \P_{\x_1}(\x_1-\x_2) \|_2^2 &\leq& \beta  \|d_{\x_1}(\x_1)-d_{\x_1}(\x_2))) \|_2^2.
 \end{eqnarray}
 Thus if $d_{\x^\star}$ is Lipschitz for all $\x^\star\in\A$ with a small constant $K$, then the condition
 \begin{equation}
\|\P(\x_1)-\P(\x_2) - \P_{\x_1}(\x_1-\x_2) \|_2^2 \leq C \|\x_1-\x_2\|_2^2, 
\end{equation}
in Corollary \ref{col:1} holds with $C=\beta K$.

Thus it remains to show that $d_{\x_1}(\x_1)$ is Lipschitz. If the Jacobian $D_{\x_1}(\x^\star)$ of $d_{\x_1}(\x_1)$ satisfies $\|D_{\x_1}(\x+t\h)\|\leq M$ for all $0\leq t\leq 1$, then we know that
\begin{equation}
\| d_{\x_1}(\x+\h) - d_{\x_1}(\x) \| \leq M\|\h\|,
\end{equation}
so that 
\begin{equation}
\| d_{\x_1}(\x_1) - d_{\x_1}(\x_2) \| =\| d_{\x_1}(\x_2+\x_1-\x_2) - d_{\x_1}(\x_2) \| \leq M\|\x_1-\x_2\|,
\end{equation}
holds if 
\begin{equation}
\|D_{\x_1}(\x_2+t(\x_1-\x_2))\|\leq M
\end{equation}
for all $0\leq t\leq 1$.

For our simple example, we see that $D_{\x_1}$ is in fact a diagonal matrix with entries $\{D_{\x_1}(\x^\star)\}_{i,i}=h'(x_i)-h'_i(\x^\star)$, where $h'_i(\x_1)=h'(\{\x_{1}\}_i)$, so that 
\begin{equation}
\{D_{\x_1}(\x_2+t(\x_1-\x_2))\}_{i,i} = h'(\{\x_2+t(\x_1-\x_2)\}_i)-h'(\{\x_{1}\}_i).
\end{equation}
Thus if $|h'(\cdot)|$ is bounded, that is, if $|h'(\cdot)|\leq M$, then  $\|D_{\x_1}(\x)\|\leq M$.

We thus have demonstrated the following.
\begin{lem}
Let $\P(\x)=\overline{\P}f(\x)$, where the function $f(\cdot)=x+h(x)$ is applied element wise and where the derivative $h'(x)$ is absolutely bounded $|h'(\cdot)|\leq M$, then 
 \begin{equation}
 \|\P(\x_1)-\P(\x_2) - \P_{\x_1}(\x_1-\x_2) \|_2^2 \leq  C\|\x_1-\x_2\|,
\end{equation}
where $C=\beta M$.
\end{lem}

%Note that the approximation error in iteration $n$ is $\|\y-\P_\x^n\x_A \|\leq  \|\y-\P(\x_A)\| +\|P_\x^n\x_A -\P(\x_A)\| $
% P(\x^n) + P_{\x^n} (\x_A-\x^n) 
%
% \|\y-\P(\x_A)\|=\|\y-\P_\x^n\x_A + \x^n\x_A- \P(\x_A)\|

\subsection{Proof of Theorem \ref{thm:main1}}
\label{proof:1}
The proof follows basically that in \cite{blumensath10subsp}, but with some important modifications to account for the non-linear setting analysed here.
\begin{proof}
As always, we start with the triangle inequality
\begin{equation}
\label{equation:chain1}
\|\x-\x^{n+1}\|\leq \|\x_{\A}-\x^{n+1}\| + \|\x_{\A}-\x\|
\end{equation}
and then bound the first term on the left using the definition 
\begin{equation}
\e_\A^{n}=\y-\P(\x^n)-\P_{\x^n}( \x_\A -\x^n )
\end{equation}
 and the  inequalities
\begin{eqnarray}
& &\|\x_{\A}-\x^{n+1}\|^2 \nonumber \\ 
&\leq& \frac{1}{\alpha}  \|\P_{\x^{n}}(\x_{\A}-\x^{n+1})\|^2 \nonumber \\ 
& = & \frac{1}{\alpha}   \|\y - \P(\x^n) - \P_{\x^{n}}(\x^{n+1}-\x^n) - (\y-\P(\x^n)-\P_{\x^n}( \x_\A -\x^n )  ) \|^2 \nonumber \\ 
& = & \frac{1}{\alpha} \|\y - \P(\x^n) - \P_{\x^{n}}(\x^{n+1}-\x^n) -\e_\A^{n} \|^2 \nonumber \\ 
& = &\frac{1}{\alpha} \left(\|\y - \P(\x^n) - \P_{\x^{n}}(\x^{n+1}-\x^n) \|^2 + \|\e_\A^{n}\|^2 -2\langle\e_\A^{n},(\y- \P(\x^n) - \P_{\x^{n}}(\x^{n+1}-\x^n))\rangle \right)\nonumber \\ 
&\leq& \frac{1}{\alpha} \left(\|\y- \P(\x^n) - \P_{\x^{n}}(\x^{n+1}-\x^n)\|^2 + \|\e_\A^{n}\|^2 +  \|\e_\A^{n}\|^2 + \|\y- \P(\x^n) - \P_{\x^{n}}(\x^{n+1}-\x^n)\|^2 \right) \nonumber \\
& = & \frac{2}{\alpha}  \|\y- \P(\x^n) - \P_{\x^{n}}(\x^{n+1}-\x^n)\|^2 +  \frac{2}{\alpha}  \|\e_\A^{n}\|^2.
\label{equation:chain2}
\end{eqnarray}
We here used the fact that 
\begin{eqnarray}
&&-\langle \e_\A^{n},(\y- \P(\x^n) - \P_{\x^{n}}(\x^{n+1}-\x^n))\rangle \nonumber\\
&\leq& \|\e_\A^{n}\| \|\y- \P(\x^n) - \P_{\x^{n}}(\x^{n+1}-\x^n)\| \nonumber \\
&\leq& 0.5(\|\e_\A^{n}\|^2+\|(\y- \P(\x^n) - \P_{\x^{n}}(\x^{n+1}-\x^n)\|^2). \nonumber
\end{eqnarray}

The left term in the last line of \eqref{equation:chain2} is bounded by the next inequality
\begin{equation}
\label{equation:chain3}
\|\y-\P(\x^n)-\P_{\x^{n}}(\x^{n+1}-\x^n)\|^2\leq (\frac{1}{\mu}-\alpha)\|(\x_\A-\x^n)\|^2 + \|\e_\A^{n}\|^2 ,
\end{equation}
which is a result of the following argument in which we use $\g = 2\P_{\x^{n}}^*(\y-\P(\x^{n}) )$
\begin{eqnarray}
& &  \| \y-\P(\x^n)-\P_{\x^{n}}(\x^{n+1}-\x^n)\|^2 - \|\y-\P(\x^n)\|^2\nonumber \\
& =   & -\langle(\x^{n+1}-\x^n), \g\rangle + \|\P_{\x^{n}}(\x^{n+1}-\x^n) \|^2 	\nonumber \\
&\leq & -\frac{2}{\mu}\langle(\x^{n+1}-\x^n), \frac{\mu}{2}\g\rangle + \frac{1}{\mu}\|(\x^{n+1}-\x^n) \|^2 	\nonumber \\
%& =   & -\langle\x^{n+1},\g\rangle + \langle\x^{n},\g\rangle + \frac{1}{\mu}\|(\x^{n+1}-\x^n) \|^2 \nonumber \\
& =   & \frac{1}{\mu}\left[ \|\x^{n+1}-\x^n-\frac{\mu}{2}\g\|^2 -\frac{\mu}{2}\|\g\|^2  \right]\nonumber \\
& = 	& \frac{1}{\mu} \left[\inf_{\x\in\A} \|\x-\x^n-\frac{\mu}{2}\g\|^2 -\frac{\mu}{2}\|\g\|^2  \right]  				\nonumber \\
& = 	& \inf_{\x\in\A} \left[  -\langle(\x-\x^n),\g\rangle + \frac{1}{\mu}\|(\x-\x^n) \|^2   \right]  \nonumber \\
&\leq	& -\langle(\x_\A-\x^n),\g\rangle + \frac{1}{\mu}\|(\x_\A-\x^n) \|^2 	 \nonumber \\
& =   & -2\langle(\x_\A-\x^n),\P_{\x^{n}}^*(\y-\P(\x^n))\rangle + \frac{1}{\mu}\|\x_\A-\x^n \|^2 	 \nonumber \\
& =   & -2\langle(\x_\A-\x^n),\P_{\x^{n}}^*(\y-\P(\x^n))\rangle + \alpha \|\x_\A-\x^n \|^2 		  \nonumber \\
& 		& +(\frac{1}{\mu}-\alpha)\|\x_\A-\x^n \|^2 							 \nonumber \\ 
&\leq & -2\langle(\P_{\x^{n}}(\x_\A-\x^n)),(\y-\P(\x^n))\rangle + \|\P_{\x^{n}}(\x_\A-\x^n) \|^2  \nonumber \\
& 		& +(\frac{1}{\mu}-\alpha)\|\x_\A-\x^n \|^2 									\nonumber \\ 
& =   & \| \y-\P(\x^n)-\P_{\x^{n}}(\x_\A-\x^n)\|^2-\|\y-\P(\x^{n})\|^2 \nonumber\\ & & +(\frac{1}{\mu}-\alpha)\|\x_\A-\x^n \|^2 	\nonumber \\ 
& =   & \|\e^{n}_\A\|^2-\|\y-\P(\x^{n})\|^2 + (\frac{1}{\mu}-\alpha)\|(\x_\A-\x^n) \|^2 	,
\end{eqnarray}
where the first and last inequalities are due to the RIP property of $\P_{\x^n}$ and the choice of $\beta\leq\frac{1}{\mu}$, whilst the second inequality is due to the fact that $\x_\A\in\A$.

We have thus shown that
\begin{equation}
\label{ITERBOUND}
\|\x_{\A}-\x^{n+1}\|^2 
\leq 2\left( \frac{1}{\mu\alpha}-1 \right)\|(\x_\A-\x^n) \|^2 + \frac{4}{\alpha}\|\e_\A^{n}\|^2. 
\end{equation}
We can now iterate the above expression. Using  $\epsilon^k = b \sum_{n=0}^{k-1} \|\y-\P_{\x^n}\x_\A\|^2 a^{k-1-n} $, where $b=4/\alpha$ and $a=2/(\mu\alpha)-2$, we get
\begin{equation}
\label{equation:chain4}
\|\x_{\A}-\x^{k}\|^2 
\leq \left(2\left(\frac{1}{\mu\alpha}-1\right)\right)^k\|\x_\A \|^2 + \epsilon^k. 
\end{equation}

Thus
\begin{eqnarray}
\|\x-\x^{k}\| &\leq& \sqrt{c^k\|\x_\A \|^2 + \epsilon^k + \|\x_{\A}-\x\|} \nonumber \\
&\leq& c^{k/2} \|\x_\A \| + \sqrt{\epsilon^k}+  \|\x_{\A}-\x\|, \nonumber
\end{eqnarray}
where $c=\frac{2}{\mu\alpha}-2$ and the theorem is proven.
%This means that after $k^\star = \left\lceil 2 \frac{\ln(\delta\frac{\|\e_\A\|}{\|\x_\A\|})}{\ln(2/(\mu\alpha)-2)} \right\rceil$ iterations we have
%\begin{equation}
%\|\x-\x^{k^\star}\| \leq (c^{0.5}+\delta)\|\e_\A\| + \|\x_{\A}-\x\| +\sqrt{\frac{c\epsilon}{2\mu}},
%\end{equation}
%which is the bound of the theorem.
\end{proof}

\subsection{Proof of Corollary \ref{col:1}}
\label{proof:2}
\begin{proof}
Let us start with the bound in  \eqref{ITERBOUND}
\begin{equation}
\|\x_{\A}-\x^{n+1}\|^2 
\leq 2\left( \frac{1}{\mu\alpha}-1 \right)\|(\x_\A-\x^n) \|^2 + \frac{4}{\alpha}\|\e_\A^{n}\|^2
\end{equation}
and let us look a bit more closely at
\begin{eqnarray}
\e_\A^{n}&=&\y-\P(\x^n)-\P_{\x^n}( \x_\A -\x^n )\nonumber \\
&=&\P(\x_\A)+\e_\A-\P(\x^n)-\P_{\x^n}( \x_\A -\x^n ),
\end{eqnarray}
where $\e_\A=\y-\P(\x_\A)$.
%To guarantee that $\e_\A^{n}$ is bounded, let us expand $\P(\x_\A)$ using our affine approximation. 
%and let $\e_{\x_\A-\x^n}$ denote the approximation error. 
We then have
\begin{eqnarray}
\|\e_\A^{n}\| & \leq&\|\P(\x_\A)-\P(\x^n)-\P_{\x^n}( \x_\A -\x^n )\| + \|\e_\A\|%\nonumber\\
%&=& \|\P(\x^n)  + \P_{\x^n}(\x_\A-\x^n) -\P(\x^n)-\P_{\x^n}( \x_\A -\x^n )+\e_{\x_\A-\x^n}\|  + \|\e_\A\| \nonumber \\
%&=& \|\e_{\x_\A-\x^n}\|  + \|\e_\A\|.
\end{eqnarray}
Now by assumption, $\|\P(\x_\A)-\P(\x^n)-\P_{\x^n}( \x_\A -\x^n )\|$ is bounded as a function of $\|\x_\A-\x^n\|$, i.e. 
\begin{equation}
\|\P(\x_\A)-\P(\x^n)-\P_{\x^n}( \x_\A -\x^n )\|^2\leq C \|\x_\A-\x^n\|^2,  
\end{equation}
so that \eqref{ITERBOUND} becomes
\begin{eqnarray}
\|\x_{\A}-\x^{n+1}\|^2 
\leq 2\left( \frac{1}{\mu\alpha}-1 + \frac{4}{\alpha} C \right)\|(\x_\A-\x^n) \|^2 + \frac{8}{\alpha} \|\e_\A\|^2,
\end{eqnarray}
Thus we require that $\frac{1}{\mu\alpha}-1 + \frac{4}{\alpha} C \leq0.5$, that is that
$1/\mu\leq1.5\alpha-4C$.

The same argument used in the main proof now holds. Whenever the constant before the left term on the right hand side is smaller than one, then we can iterate the error and the corollary follows.
\end{proof}

\section{The Iterative Hard Thresholding Algorithm for Non-Linear Optimisation}
\label{subsection:algo}
%NOTE: SUBGRADIENT REQUIRES CONVEXITY, WHICH IS TOO RESTRICTIVE, SO WE ARE BETTER ASSUMING DIFFERENTIABILITY AND ASSUME $\nabla (\x_1)$ TO BE THE GRADIENT! WE MIGHT BE ABLE TO ALLOW NON_DIFFERENTIABLE BITS AND ASUME LOCAL SUB_GRADIENTS.

Let us now return to the more general problem of minimising a non-linear function  $f(\x)$ under the constraint that $\x\in\A$, where $\A$ is a Union of Subspaces.
Let us recall again that for minimisation problems of the form $\mathrm{argmin}_{\x\in\A} \|\y-\P\x\|^2$ we use the algorithm
\begin{equation}
\x^{n+1}=P_\A(\x^n+\P^*(\y-\P\x^{n}).
\end{equation}
Note that the update $\P^*(\y-\P\x^{n})$ is a scaled version of the gradient of the cost function $\|\y-\P\x\|^2$.

In the more general setting $\mathrm{argmin}_{\x\in\X} f(\x)$, where $\x$ is an Euclidean vector, we can simply replace this update direction with the gradient of $f(\x)$ (evalustaed at $\x^n$), whilst in more general spaces, we assume that $f(\x)$ is Fr\'echet  differentiable with respect to $\x$, that is, for each $\x_1$ there exist a linear functional $D_{\x_1}(\cdot)$ such that
\begin{equation}
\lim_{\h\rightarrow0} \frac{ f(\x_1+\h) - f(\x_1) - D_{\x_1}(\h)}{\|\h\|}=0.
\end{equation}
We can then use Riesz representation theorem to write the linear functional $D_{\x_1}(\cdot)$ using its inner product equivalent
\begin{equation}
D_{\x_1}(\cdot) =  \langle \nabla (\x_1), \cdot \rangle,
\end{equation}
where $\nabla (\x_1)\in\H$.
Using $\h=\x_2-\x_1$ we see that for each $\u$ and $\x_1$ we require the existence of a $\nabla (\x_1)$ such that
\begin{equation}
\label{equation:gradientdef}
\lim_{\x_2\rightarrow\x_1} \frac{ f(\x_2) - f(\x_1) - \langle \nabla (\x_1), (\x_2-\x_1)\rangle } {\|\x_2-\x_1\|_\H}=0.
\end{equation}

In Euclidean spaces the Fr\'echet derivative is obviously the differential of $f(\x)$ at $\x_1$, in which case $\nabla (\x_1)$ is the gradient and $\langle \cdot,\cdot\rangle$ the Euclidean inner product. With a slight abuse of terminology, we will therefore call $\nabla (\x_1)$ `the gradient' even in more general Hilbert space settings.%\footnote{Note that in general, for our main result to hold, we don't even require that $\|\y-\P(\x)\|_B^2$ is Fr\`echet differentiable. All that is required for our theorem to hold is that for each $\x\in\A$ there exist a $\nabla (\x_1)$ such that the Restricted Strong Convexity Property (defined below) holds.}.
%has a subgradient with respect to $\x$, that is, there is an element $\nabla (\x)\in H$ that, for fixed $\y$ and arbitrary $\x_1$ and $\x_2$ satisfies the condition 
%\begin{equation}
%Re\langle \nabla (\x_1), \x_2 \rangle +\|\y-\P(\x_1)\|_B^2 \leq \|\y-\P(\x_1+\x_2)\|_B^2.
%\end{equation}

Having thus defined an update direction $\nabla (\x)$ in quite general spaces, we are now in a position to define an algorithmic strategy to optimise $f(\x)$.
We again use a version of our trusty Iterative Hard Thresholding algorithm, but replace the update direction with $\nabla (\x)$. With this modification, the algorithm might also be called the Projected Landweber Algorithm \cite{book00reg}, and is defined formally by the iteration
\begin{equation}
\x^{n+1} = P_{\A} (\x^n - (\mu/2)  \nabla (\x^n)),
\end{equation}
where $\x^0=\mathbf{0}$ and $\mu$ is a step size parameter chosen to satisfy the condition in Theorem \ref{thm:main2} below.

\subsection{Theoretical Performance Bound}
%%%%%%%%%%%%%

We now come to the second  main result of this paper, which states that, if $f(\x)$ satisfy the Restricted Strong Convexity Property, then the Iterative Hard Thresholding algorithm can find a vector $\x$ that is close to the true minimiser of $f(\x)$ among all $\x\in\A$.
In particular, we have the following theorem.
\begin{thm}
\label{thm:main2}
Let $\A$ be a union of subspaces. 
Given the optimisation problem $f(\x)$, where $f(\cdot)$ is a positive function that satisfies the  \emph{Restricted Strict Convexity Property} 
\begin{equation}
\alpha\leq\frac{f(\x_1) - f(\x_2) -Re\langle\nabla (\x_2),(\x_1-\x_2)\rangle }{\|\x_1-\x_2\|^2}\leq\beta,
\end{equation}
for all $\x_1,\x_2\in\H$ for which $\x_1-\x_2\in \A+\A+\A$. 
Let $\x_{opt} = \mathrm{argmin}_{\x\in\A}f(\x)$ and assume that $\beta\leq \frac{1}{\mu}\leq\frac{4}{3}\alpha$,
then, after
\begin{equation}
n^\star = 2 \frac{\ln\left(\delta \frac{f(\x_{opt})}{ \|\x_{opt}\|}\right)}{\ln{4(1-\mu\alpha)}},
\end{equation}
iterations, the IHT algorithm calculates a solution $\x^{n^\star}$ satisfying
\begin{eqnarray}
\|\x^{n^\star}-\x\| \leq \left(2 \sqrt{\frac{ \mu}{1-c}} +\delta\right) f(\x_{opt}) + \|\x-\x_{opt}\| .
\end{eqnarray}
\end{thm}
In the traditional Compressed Sensing setting, this result is basically that derived in \cite{blumensath09IHT}.

\subsection{Proof of the Second Main Result}

\begin{proof}[Proof of Theorem \ref{thm:main2}]

%For any $\Gamma$, let $\a^n_\Gamma$ be the projection of $\a^n$ onto the subspace $\Gamma$. 
%Let $\Gamma$ be chosen such that $P_{\A}^{\epsilon}(\x) = P_{\A}^{\epsilon}(\x_\Gamma)$ and such that $\x^n, \a_\A \in\Gamma$. This proof works for $\A$ being the union of subspaces, where $\Gamma$ is the union of any THREE subspaces, such that $\x^{n+1}$ and $\x^n$ both lie in $\Gamma$. 
The proof requires the orthogonal projection onto a subspace $\Gamma$. The subspace $\Gamma$ is defined as follows. Let $\Gamma$ be the sum of no more than three subspaces of $\A$, such that $\x_{opt},\x^n,\x^{n+1}\in\Gamma$. Let $P_\Gamma$ be the orthogonal projection onto the subspace $\Gamma$. We write $\a^n_\Gamma=P_\Gamma\a^n$ and $P_\Gamma\nabla(\x^n) = \nabla_\Gamma(\x^n)$. Note that this ensures that $P_\Gamma\x^{n}=\x^{n}$, $P_\Gamma\x^{n+1} =  \x^{n+1}$ and $P_\Gamma \x_{opt} =  \x_{opt}$. 

We note for later that with this notation
\begin{eqnarray}
Re\langle \nabla_\Gamma(\x^n),(\x_{opt}-\x^{n}) \rangle &=& Re\langle P_\Gamma\nabla(\x^n),(\x_{opt}-\x^{n}) \rangle \nonumber\\
&=& Re\langle \nabla(\x^n),P_\Gamma(\x_{opt}-\x^{n}) \rangle \nonumber\\
&=& Re\langle \nabla(\x^n),(\x_{opt}-\x^{n}) \rangle
\end{eqnarray}
and
\begin{eqnarray}
\|\nabla_\Gamma(\x^n)\|^2=\langle \nabla_\Gamma(\x^n),\nabla_\Gamma(\x^n)\rangle 
&=& \langle  P_\Gamma\nabla(\x^n),P_\Gamma\nabla(\x^n)\rangle \nonumber \\
&=& \langle  \nabla(\x^n),P_\Gamma^*P_\Gamma\nabla(\x^n)\rangle \nonumber \\
&=& \langle  \nabla(\x^n),\nabla_\Gamma(\x^n)\rangle.
\end{eqnarray}

We also need the following lemma.
\begin{lem}
\label{lem:gradbound}
Under the assumptions of the theorem,
\begin{equation}
\| \frac{\mu}{2}\nabla_\Gamma(\x^n)\|^2 - \mu f(\x^n)\leq 0.
\end{equation}
\end{lem}
\begin{proof}
Using the \emph{Restricted Strict Convexity Property} we have
\begin{eqnarray}
\| \frac{\mu}{2}\nabla_\Gamma(\x^n)\|^2 
&=& - \frac{\mu}{2}Re\langle \nabla(\x^n),-\frac{\mu}{2}\nabla_\Gamma(\x^n)\rangle \nonumber \\
& \leq & \frac{\mu}{2} \beta \|\frac{\mu}{2}\nabla_\Gamma(\x^n)\|^2 + \frac{\mu}{2}f(\x^n) - \frac{\mu}{2} f(x^n-\frac{\mu}{2}\nabla_\Gamma(\x^n))) \nonumber\\
& \leq & \frac{\mu}{2} \beta \|\frac{\mu}{2}\nabla_\Gamma(\x^n)\|^2 + \frac{\mu}{2}f(\x^n). 
\end{eqnarray}
Thus
\begin{eqnarray}
(2-\mu\beta)\|\frac{\mu}{2}\nabla_\Gamma(\x^n)\|^2 & \leq &  \mu f(\x^n), 
\end{eqnarray}
which is the desired result as $\mu\beta\leq 1$ by assumption.
\end{proof}

To prove the theorem, we start by bounding the distance between the current estimate $\x^{n+1}$ and the optimal estimate $\x_{opt}$. Let $\a^n_\Gamma=\x^{n}_\Gamma-\mu/2\nabla_\Gamma(\x^n)$. Because $\x^{n+1}$ is the closest element in $\A$ to $\a^n_\Gamma$, we have
\begin{eqnarray}
\|\x^{n+1}-\x_{opt}\|^2 
& \leq & \left(\|\x^{n+1} - \a^n_\Gamma\|+\|\a^n_\Gamma-\x_{opt}\| \right)^2 \nonumber \\
&\leq&  4\|(\a^n_\Gamma-\x_{opt})\|^2 \nonumber \\
& =&  4 \|\x^{n}-(\mu/2)\nabla_\Gamma(\x^n)-\x_{opt}\|^2    \nonumber \\ 
& = &4 \|(\mu/2)\nabla_\Gamma(\x^n)+(\x_{opt}-\x^{n})\|^2 \nonumber \\
& = & \mu^2 \|\nabla_\Gamma(\x^n)\|^2 + 4 \|\x_{opt}-\x^{n}\|^2 +  
4\mu Re\langle \nabla_\Gamma(\x^n),(\x_{opt}-\x^{n}) \rangle  \nonumber \\
& = & \mu^2 \|\nabla_\Gamma(\x^n)\|^2 + 4 \|\x_{opt}-\x^{n}\|^2 +  
4\mu Re\langle \nabla(\x^n),(\x_{opt}-\x^{n}) \rangle   \nonumber \\
& \leq &4   \|\x_{opt}-\x^{n}\|^2 + \mu^2 \|\nabla_\Gamma(\x^n)\|^2    \nonumber \\
& & +4\mu [-\alpha\|\x^n-\x_{opt}\|^2 +f(\x_{opt})-f(\x^n)]   \nonumber \\
& = & 4(1-\mu\alpha) \|\x_{opt}-\x^{n}\|^2 +4\mu f(\x_{opt})    \nonumber \\
& & + 4  [\|(\mu/2)\nabla_\Gamma(\x^n)\|^2 - \mu f(\x^n)]\nonumber \\
& \leq & 4(1-\mu\alpha) \|\x_{opt}-\x^{n}\|^2 +4\mu f(\x_{opt})  .
\end{eqnarray}
Here, the second to last inequality is the RSCP and the last inequality is due to lemma \ref{lem:gradbound}.

We have thus shown that
\begin{eqnarray}
\|\x^{n+1}-\x_{opt}\|^2 \leq 4(1-\mu\alpha) \|\x_{opt}-\x^{n}\|^2 + 4\mu f(\x_{opt}) .
\end{eqnarray}
Thus, with $c=4(1-\mu\alpha)$
\begin{eqnarray}
\|\x^{n}-\x_{opt}\|^2 \leq c^n \|\x_{opt}\|^2 + \frac{ 4\mu}{1-c} f(\x_{opt})  ,
\end{eqnarray}
so that, if $\frac{1}{\mu}<\frac{4}{3}\alpha$ we have $c=4(1-\mu\alpha)<1$, so that $c^n$ decreases with $n$.
Taking the square root on both sides and noting that for positive $a$ and $b$, $\sqrt{a^2+b^2}\leq a + b$,
\begin{eqnarray}
\|\x^{n}-\x_{opt}\| \leq c^{n/2} \|\x_{opt}\| + 2 \sqrt{\frac{ \mu}{1-c}} f(\x_{opt})  .
\end{eqnarray}

The theorem then follows using the triangle inequality
\begin{eqnarray}
\|\x^{n}-\x\| &\leq& \|\x^{n}-\x_{opt}\| + \|\x-\x_{opt}\|\nonumber \\
 & \leq & c^{n/2} \|\x_{opt}\| + 2 \sqrt{\frac{ \mu}{1-c}} f(\x_{opt}) \nonumber \\
 & &+ \|\x-\x_{opt}\|.
\end{eqnarray}
The iteration count is found by setting 
\begin{equation}
c^{n/2} \|\x_{opt}\| \leq \delta (\x_{opt}).
\end{equation}
so that after
\begin{equation}
n = 2 \frac{\ln\left(\delta \frac{f(\x_{opt})}{ \|\x_{opt}\|}\right)}{\ln{c}},
\end{equation}
iterations
\begin{eqnarray}
\|\x^n-\x\| \leq \left(2 \sqrt{\frac{ \mu}{1-c}} +\delta\right) f(\x_{opt}) + \|\x-\x_{opt}\|   .
\end{eqnarray}
%Where we use the assumption from the Theorem which implies that $\frac{1}{\mu\alpha}\leq\frac{4}{3}$.
\end{proof}

%%%%%%%%%%%%%%%%%%%%%%%%%%%%%%%%%%

\subsection{When and Where is this Theory Applicable?}

Since we first derived the result here, it has been shown that properties such as the Restricted Strict Convexity Property do indeed hold for certain non-linear functions such as those encountered in certain logistic regression problems \cite{Bahmani12nonlin}. These recent findings thus further strengthen the case for a detailed study of non-convexly constrained non-linear problems and the derivation of novel methodologies for their solution. %Our result should therefore be seen in this light, it does however not directly translate into a simple application to Compressed Sensing under non-linear observations.

It may thus seem tempting to use this theory also in a non-linear Compressed Sensing setting, where we would have $f(\x)=\|\y-\P(\x)\|_B^2$, where $\|\cdot\|_\B$ is some Banach space norm and where $\P(\cdot)$ is some non-linear function\footnote{This was indeed the setting proposed in  \cite{blumensath10nonlin}.}. If this $f(\x)$ would satisfy the Restricted Strict Convexity property, then the Theory in the second part of this paper would indeed tell us how to solve the non-linear Compressed Sensing problem.

Unfortunately, it is far from clear yet under which conditions on $f(\x)=\|\y-\P(\x)\|_B^2$ Restricted Strict Convexity type properties hold. 
Indeed, the following lemma shows that such a condition cannot be fulfilled in general for Hilbert spaces.
\begin{lem}
Assume $\B$ is a Hilbert space and assume $f(\x)$ is convex on $\A+\A$ for all $\y$ (i.e. it Satisfies the Restricted Strict Convexity Property), then $\P$ is affine on all subspaces of $\A+\A$.
\end{lem}
\begin{proof}
The proof was suggested by an anonymous reviewer of the earlier version of this manuscript \cite{blumensath10nonlin} and uses contradiction. Assume $\P$ is not affine on any subspace of $\A+\A$. Thus, there is a subspace $\S=\A_i+\A_j$, and $\x_n\in\S$, such that for $\x=\sum_n \lambda_n\x_n$, where $\sum_n\lambda_n=1$ and $0\leq \lambda_n$, we have $\sum_n\P(\x_n)-\P(\x)\neq \mathbf{0}$.
Now by assumption of strong convexity on $\S$, we have (using $\y_n=\P(\x_n)$ and $-\bar{\y}=\x$)
\begin{eqnarray}
\mathbf{0}\leq \sum_n\lambda_n \|\y-\P(\x_n)\|^2-\|\y-\P(\x)\|^2= \sum_n\lambda_n \|\y-\y_n\|^2-\|\y-\bar{\y}\|^2 \nonumber \\
= 2\langle \y ,\bar\y- \sum_n \lambda_n\y_n  \rangle + \sum_n\lambda_n \|\y_n\|^2-\|\bar\y\|^2. 
\end{eqnarray}
where the inequality is due to the assumption of convexity. 
But the above inequality cannot hold for all $\y$ (it fails for example for a multiple of $-(\bar\y- \sum_n\lambda_n\y_n)$). Thus $\P$ needs to be affine on the linear subsets of $\A+\A$. 
\end{proof}
Whilst this implies that the property cannot hold in Hilbert spaces for non-affine $\P$ and \emph{all} $\y$, it does not preclude the possibility that it could hold for specific observations $\y$. This would not allow us to build a general signal recovery framework, but might still allow us the recovery of a subset of signals. Thus, for the non-linear Compressed Sensing problem in Hilbert space, the Restricted Isometry Property of the Jacobian of $\P(\x)$ together with the ability to construct a good linear approximation of  $\P(\x)$ seem to be the more suitable tools  to study recovery performance. Nevertheless, for certain other non-convexly constrained non-linear optimisation problems, such as those addressed in \cite{Bahmani12nonlin}, the Restricted Strict Convexity Property might be the more appropriate framework.  Whilst there are many similarities between these requirements and they both boil down to the same RIP property in the linear setting, it remains to be seen what the exact relationship is between these two measures in general non-linear problems.

%Nevertheless, the linerise IHT algorithm offers a powerful tool to tackle Compressed Sensing under a non-linear observation model and we leave a more detailed study of the IHT algorithm of this section and its use for Compressed Sensing for future work.

%%%%%%%%%%%%%%%%%%%%%%%%%%%%%%%%%%%%%%%%%%%%%%%%%%%%
%%%%%%%%%%%%%%%%%%%%%%%%%%%%%%%%%%%%%%%%%%%%%%%%%%%%
%%%%%%%%%%%%%%%%%%%%%%%%%%%%%%%%%%%%%%%%%%%%%%%%%%%%
%%%%%%%%%%%%%%%%%%%%%%%%%%%%%%%%%%%%%%%%%%%%%%%%%%%%
%%%%%%%%%%%%%%%%%%%%%%%%%%%%%%%%%%%%%%%%%%%%%%%%%%%%
%%%%%%%%%%%%%%%%%%%%%%%%%%%%%%%%%%%%%%%%%%%%%%%%%%%%
%%%%%%%%%%%%%%%%%%%%%%%%%%%%%%%%%%%%%%%%%%%%%%%%%%%%
%%%%%%%%%%%%%%%%%%%%%%%%%%%%%%%%%%%%%%%%%%%%%%%%%%%%
\section{Conclusions}

Compressed Sensing ideas can be developed in much more general settings than considered traditionally. We have shown previously \cite{blumensath10subsp} that sparsity is not the only structure that allows signals to be recovered and that the finite dimensional setting can be replaced with a much more general Hilbert space framework. 
In this paper we have made a further important generalisation and have introduced the concept of non-linear measurements into Compressed Sensing theory. Under certain conditions, such as the requirement that the Jacobian of the measurement system satisfies a Restricted Isometry Property, then the Iterative Hard Thresholding algorithm can be used to recover signals from a non-convex constraint set with similar error bounds to those derived in Compressed Sensing.

In the second part of this paper we have then looked at the related and in some sense more general setting of non-linear optimisation under non-convex constraints. Here we have looked the Restricted Strict Convexity Property as a tool to study recovery performance and it was shown that that this condition is indeed sufficient for the Iterative Hard Thresholding to find points that are near the optimal solution.

%We have here developed the main result based on ideas from \cite{blumensath09IHT} and \cite{blumensath10subsp}, however, it appears reasonable to assume that a similar result also holds for more general sets $\A$ and for a RSCP with respect to $\A+\A$ rather than $\A+\A+\A$. %We are currently working towards such a more general result based on the approaches in \cite{blumensath10inv} and \cite{fucart10sparse}.

%%%%%%%%%%%%%%%%%%%%%%%%%%%%%%%%%%%%%%%%%%%%%%%%%%%%
%%%%%%%%%%%%%%%%%%%%%%%%%%%%%%%%%%%%%%%%%%%%%%%%%%%%
\section*{Acknowledgment}
\noindent This work was supported in part by the UKÕs Engineering and Physical Science Research Council grants EP/J005444/1 and D000246/1 and a Research Fellowship from the School of Mathematics at the University of Southampton.

%\bibliographystyle{IEEEbib}
%\bibliography{//samba.soton.ac.uk/tb1m08/Documents/Bibliography/ThomasBibliography3}

\end{document}